\documentclass[12pt, a4paper]{article}
\usepackage{graphicx,amssymb,amsmath}

\usepackage{fullpage}
\usepackage{boxedminipage}
\usepackage{url}

\usepackage{authblk}
\usepackage{hyperref}
\hypersetup{
	hidelinks,
	colorlinks=true,
	citecolor=[rgb]{0.121 0.47 0.705},
	linkcolor=[rgb]{0.121 0.47 0.705},
	urlcolor=[rgb]{0.121 0.47 0.705}
}

\usepackage{xcolor}

\newcommand{\msize}[1]{\left|#1\right|}
\DeclareMathOperator{\closure}{cl}

\usepackage{amsthm}      
\newtheorem{theorem}{Theorem}
\newtheorem{lemma}[theorem]{Lemma}

\newtheorem{obs}[theorem]{Observation}

\title{Efficient Segment Folding is Hard\thanks{This work was started during the FWF--JSPS Japan--Austria Bilateral Seminar ``Computational Geometry Seminar with Applications to Sensor Networks" held in Zao (Japan) in 2018. The authors want to thank all participants for the fruitful discussions.
We want to thank the anonymous reviewers for the insightful comments.}}
\author[1]{Takashi~Horiyama\thanks{Email: horiyama@ist.hokudai.ac.jp}}
\author[2]{Fabian~Klute\thanks{Email: f.m.klute@uu.nl. Supported by the Netherlands Organisation for Scientific Research (NWO) under project no. 612.001.651 and by the Austrian Science Fund (FWF): J-4510.}}
\author[3]{Matias~Korman\thanks{Email: matias.korman\_cozzetti@siemens.com.}}
\author[4]{Irene~Parada\thanks{Email: irmde@dtu.dk. Research partially supported by the Austrian Science Fund (FWF): W1230 and by Independent Research Fund Denmark grant 2020-2023 (9131-00044B) ``Dynamic Network Analysis''.}}
\author[5]{Ryuhei~Uehara\thanks{Email: uehara@jaist.ac.jp. Supported by MEXT Kakenhi No.~17H06287 and 18H04091.}}
\author[6]{Katsuhisa~Yamanaka\thanks{Email: yamanaka@cis.iwate-u.ac.jp}}

\affil[1]{Faculty of Information Science and Technology, Hokkaido University, Sapporo, Japan}
\affil[2]{Utrecht University, Utrecht, The Netherlands}
\affil[3]{Siemens EDA (formerly Mentor Graphics), OR, USA}
\affil[4]{Technical University of Denmark, Lyngby, Denmark}
\affil[5]{School of Information Science, JAIST, Ishikawa, Japan}
\affil[6]{Faculty of Science and Engineering, Iwate University, Iwate, Japan}

\date{}

\begin{document}
	
	\maketitle

\begin{abstract}
	We introduce a computational origami problem which we call the {\em segment folding} problem: 
	given a set of $n$ line-segments in the plane the aim is to make creases along all segments in the minimum number of folding steps. 
	Note that a folding might alter the relative position between the segments, 
	and a segment could split into two. 
	We show that it is NP-hard to determine whether $n$ line segments can be folded in $n$ simple folding operations.
\end{abstract}

\section{Introduction}
Origami designers around the world struggle
with the problem of finding a better way to fold an origami model.
Recent advanced origami models require substantial \emph{precreasing} of
a prescribed mountain-valley pattern (getting each crease folded slightly
in the correct direction), and then folding all the creases at once.
For example, for folding the MIT seal \emph{Mens et Manus} in ``three easy steps'', Chan~\cite{Chan} spent roughly three hours precreasing, three hours folding those creases, and four hours of artistic folding.
The precreasing component is particularly tedious.
Thus, we wonder if this process can be automated by folding robots. As of the writing of this paper, the most recent robots for folding paper can achieve
only quite simple foldings (see, e.g., \cite{OriRobo}), but we consider a future setting in which more difficult ones are possible. In such a setting, we have a series of portions of the sheet that need to be folded in some specified locations, and we would like to do it as fast as possible (that is, in the minimum possible number of foldings).

We consider one of the simplest folding operations possible called \emph{all-layers simple fold}.
An all-layers simple fold starts with a sheet in a flat folded state 
and consists of the following three steps:
(1) choose a crease line, 
(2) fold all paper layers along this line, and
(3) crease to make all paper layers flat again.

In our model we start with a plain sheet of paper with no folds or creases (say, a unit square). 
From that sheet we do all-layers simple folds one at a time until we have reached our desired shape. 
We note that this is one of the simplest folding models, 
and that many variants have been considered in the literature (see \cite{ADK17} for  other folding models).

This situation leads us to the following natural \emph{segment folding problem}:
\begin{center}
\begin{boxedminipage}{0.98 \columnwidth}
\textsc{Segment Folding Problem} \\ [5pt]
\begin{tabular}{l p{0.80 \columnwidth}}
Input: & A set $S$ of line segments $s_1,s_2,\ldots,s_{|S|}$ in the plane (on a sheet of paper) and an integer $\kappa$.\\
Operation: & Reflection (all-layers simple fold) along a line containing one or more (parts of) segments $s_i$ for some $i\in \{1,\dots, |S|\}$. \\
Question: & Is there a sequence of $\kappa$ operations such that all (parts of) segments $s_i$ become reflection (crease) lines?

\end{tabular}
\end{boxedminipage}
\end{center} 
{We remark that only folds along the supporting lines of the line segments are possible. 
As we will see, this setting gives rise to interesting questions.} 
Notice that, 
when we fold along a line $\ell$, the location of all segments in one side of the line are reflected (via line symmetry) onto the other half-plane. 
In particular, if $\ell$ intersects the interior of some segment $s_j$ for $j\leq n$, it may create two segments that form a $V$ shape and meet at the folded line (it will create two segments if and only if $\ell$ and $s_j$ do not meet at a right angle and $s_j \not\subset \ell$). 
Whenever this happens, we have to fold the two subsegments since the original segment has been split into two.

\begin{figure}
	\centering
	\includegraphics{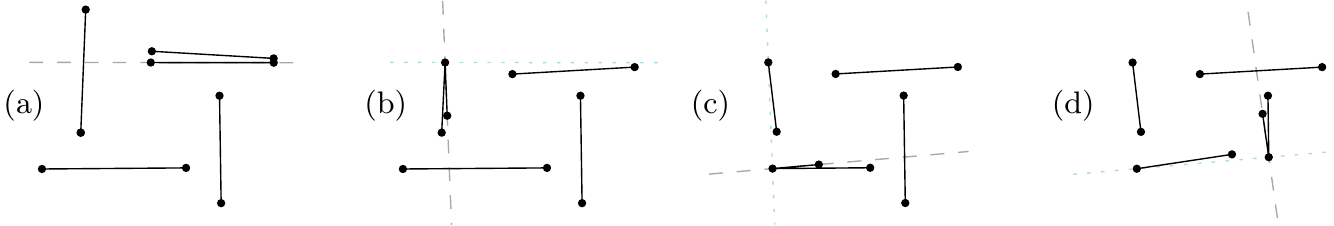}
	\caption{A problem instance that can be folded arbitrarily many times along non-folded input segments.
		If you fold along the dashed lines, the instance cycles through the four patterns and never ends. 
		Note that in all cases we have two segments forming a V shape and three almost but not exactly orthogonal segments. Each folding changes the length and the position of the V shape, but the overall structure remains the same. Note that it is possible to fold this instance with a finite number of moves (folding any of the segments that is not part of the V shape).}
	\label{fig:inf}
\end{figure}

Because folding along a line may increase the number of folds that are needed, we wonder if every instance can be folded with a finite number of fold operations\footnote{We thank Takeshi Tokuyama for posing this question.}. We say that an instance is {\em foldable} if it has a solution to the segment folding problem in a finite number of moves. 
To date, we do not know if every problem instance is foldable; in \figurename~\ref{fig:inf} we show an example admitting an infinite sequence of folds.  
Another natural question is related to efficiency: If an instance is foldable, can we find the sequence of folds? Ideally, one that minimizes that number of fold operations that are needed? 
In this paper we prove that, unless P=NP, we cannot find this optimal sequence in polynomial time.  
\begin{theorem}
	\label{thm:nphardness}
	Deciding whether a segment folding problem instance $S$ consisting of a set of $|S|$ segments
	can be solved with $|S|$ folds 
	is NP-hard.
\end{theorem}

\section{Preliminaries}
Let $S$ be a given set of $|S|$ segments in $\mathbb{R}^2$. 
Let $\ell_s$ be the supporting line of a segment $s \in S$ 
and let $R_{s}$ be the right open half-plane\footnote{For a horizontal supporting line $\ell_s$ we can define $R_{s}$ to be the top half-plane. Else, the right open half-plane is the one that contains all points $(\chi,0)$ for all values of $\chi$ larger than some threshold value. } bounded by $\ell_s$.
For a set $U\subseteq \mathbb{R}^2$ (in this paper $U$ is a collection of segments) and a segment $s\in S$, we denote as $\mathcal{R}_{s}(U)$ the reflection of set $U$ along the supporting line $\ell_s$ of segment $s$.
A fold along the supporting line $\ell_s$ of a segment $s\in S$ creates a new set of segments $S' := \closure(\{ (S \cap R_s) \cup \mathcal{R}(S \setminus \closure(R_s))\})$ 
where the segments (or parts of them) from $S$ in the left half-plane defined by $\ell_s$ are reflected into the right half-plane (or viceversa).
We denote that operation a \emph{fold along $s$}, or simply, \emph{folding~$s$}. Note that we only allow folding along segments of $S$.

We start by making straight-forward observations about basic folds.
Recall that we are interested in deciding whether an instance can be folded using at most $|S|$ folds. 
Folding along a segment $s$ on the convex hull does not change the relative positions of the remaining segments, and therefore the resulting instance after folding is the same or a reflected copy of the one resulting from just removing $s$. These two instances are equivalent for the problem we consider and therefore it is enough to just remove $s$:

\begin{obs}\label{obs:ch}
	If a segment $s$ lies on the boundary of the convex hull of our set of segments, a fold along it does not change any of the remaining segments.
\end{obs}

When a folding line supporting a segment $s\in S$ intersects the interior of another segment $s'\in S$  in a non-right angle, 
this second segment is split into two segments that share an endpoint. Both segments are now part of $S$ and they must therefore become reflection lines.
Let~$ S $ be a set of segments in $ \mathbb{R}^2 $, we say a segment $ s \in S $ \emph{stabs} $ t \in S $, $ s \neq t $, if the supporting line $ l_s $ intersects $ t $, but $ s $ and $ t $ do not share a common point.
We say a segment $ s \in S $ is stabbing, if there exists a $ t \in S $ such that~$ s $ stabs~$ t $.  

\begin{obs}\label{obs:non_stab}
	Let $S$ be a set of segments 
	in $\mathbb{R}^2$ 
	and $s \in S$ be a non-stabbing segment, 
	then if a fold along $s$ does not produce a crossing between another two segments, no segment is split. 
\end{obs}

Consider an instance in {\em general position}, that is, 
no two segments of $S$ lie on the same line  
even after we have performed up to $|S|$ folds. 
In that case each of the first $|S|$ folds can only decrease the size of $S$ by one. 

\begin{obs}\label{obs:stab}
	Given a segment folding problem instance $S$ in general position, 
	a solution sequence of length $|S|$ 
	cannot make a fold along a stabbing segment  nor a fold that makes two segments cross in a non-right angle.
\end{obs}

This claim follows from $(i)$ the fact that, in general position, folding along a stabbing segment 
does not decrease the number of segments;  
$(ii)$ two segments crossing in a non-right angle cannot be folded in less than $3$ folds, and $(iii)$ we are interested in determining if a problem instance $S$ can be solved with  at most $|S|$ folds.

\section{Reduction}\label{sec:hardness}

The reduction we first describe in this section uses for simplicity two perpendicular directions for the segments, and thus, 
constructs a set of segments~$S$ that is not in general position. 
 In the folding sequence we disallow folding along stabbing segments and along segments that produce crossings.
More precisely, the problem we first show hardness for is the \emph{restricted segment folding problem}:
\begin{center}
\begin{boxedminipage}{0.98 \columnwidth}
\textsc{Restricted Segment Folding Problem} \\ [5pt]
\begin{tabular}{l p{0.80 \columnwidth}}
Input: & A set $S$ of line segments $s_1,s_2,\ldots,s_{|S|}$ in the plane and an integer $\kappa$.\\
Operation: & Fold along a line $\ell$ containing one or more segments $s_i$ for some $i\in \{1,\dots, |S|\}$. 
The line $\ell$ must not intersect the interior of any other segment in $S$ and folding along it must not produce crossings between segments in $S$.\\
Question: & Is there a sequence of $\kappa$ operations such that all (parts of) segments $s_i$ become folding lines?
\end{tabular}
\end{boxedminipage}
\end{center}
We show the following result: 
\begin{theorem}\label{thm:restricted}
    Deciding whether a restricted segment folding problem instance $S$ 
    consisting of a set of $|S|$ segments can be solved with $|S|$ folds is strongly NP-hard, 
    even if the segments in $S$ have only three directions. 
\end{theorem}

In Section~\ref{sec:gp} we show how to modify the construction in the reduction in the proof of Theorem~\ref{thm:restricted} to show NP-hardness for the segment folding problem (proving Theorem~\ref{thm:nphardness}). 

We will prove Theorems~\ref{thm:nphardness}  and~\ref{thm:restricted} by reducing from 3SAT. 
A 3SAT formula is given by a set $ \{x_1,\dots ,x_n\} $ of boolean variables and $ \{c_1,\dots,c_m\} $ of clauses. 
Each clause contains the disjunction of three literals. 
A literal is a positive or negative occurrence of a variable. 
A formula $F(x_1,\dots, x_n)$ is the conjunction of the clauses $ c_1, \dots, c_m $. 
In 3SAT we say $F$ is satisfied if and only if for each clause $ c_j $ at least one of its literals evaluates to true. 
For the reduction we will construct a set $S$ of segments from a given 3SAT formula $F$.
To make the final correctness proof simpler, we consider only 3SAT formulas
in which no clause has only positive or negative literals. 
Transforming a given 3SAT formula in such a way is easily possible by the following folklore lemma.
For reference, a proof can be found in~\cite{DBLP:conf/wg/ArroyoKPSVW20}.
\begin{lemma}[Lemma~1 in \cite{DBLP:conf/wg/ArroyoKPSVW20}]\label{lem:transform}
	The following transformation of a clause with only positive or only negative literals, respectively, preserves the satisfiability of the clause
	($y$ is a new variable and $\mathtt{false}$ is the constant value false): 
	\begin{align*}
		x_i \! \lor \! x_j \! \lor \! x_k \Rightarrow \!
		\begin{cases}
			x_k \! \lor \!  y \lor \mathtt{false} & \! \text{(i)}\\
			x_i \! \lor \!   x_j \! \lor \! \neg y & \! \text{(ii)}
		\end{cases}
		& & \hspace{-2.75mm}
		\neg x_i \! \lor \! \neg x_j \! \lor \! \neg x_k \Rightarrow \!
		\begin{cases}
			\neg x_i  \! \lor  \! \neg x_j  \! \lor  \! y & \! \text{(iii)}\\
			\neg x_k  \! \lor  \! \neg y  \! \lor  \! \mathtt{false} & \! \text{(iv)}
		\end{cases}
	\end{align*}
\end{lemma}

Additionally to modifying the formula in the way described by Lemma~\ref{lem:transform}
we also ensure that the last clause in the formula is one containing only positive literals, i.e.,
a clause with one or two literals both positive.
If no such clause exists, we simply create a clause of size one and a new variable that is only present in this clause.
In all following sections we assume that the given 3SAT formula has been transformed as in Lemma~\ref{lem:transform} and
that the last clause has only positive literals.

To better illustrate the rather complex construction we developed a tool to procedually generate a configuration from a given 3SAT formula.
The code can be found online\footnote{\url{https://gitlab.com/fklute/segment-folding}} 
as well as a video showing some simple examples\footnote{\url{https://drive.google.com/file/d/1ythWVL1Wav-uODUFAdKG4xEXOWxDK9zx/view?usp=sharing}}.

\subsection{Overview}

\begin{figure*}
	\centering
	\includegraphics{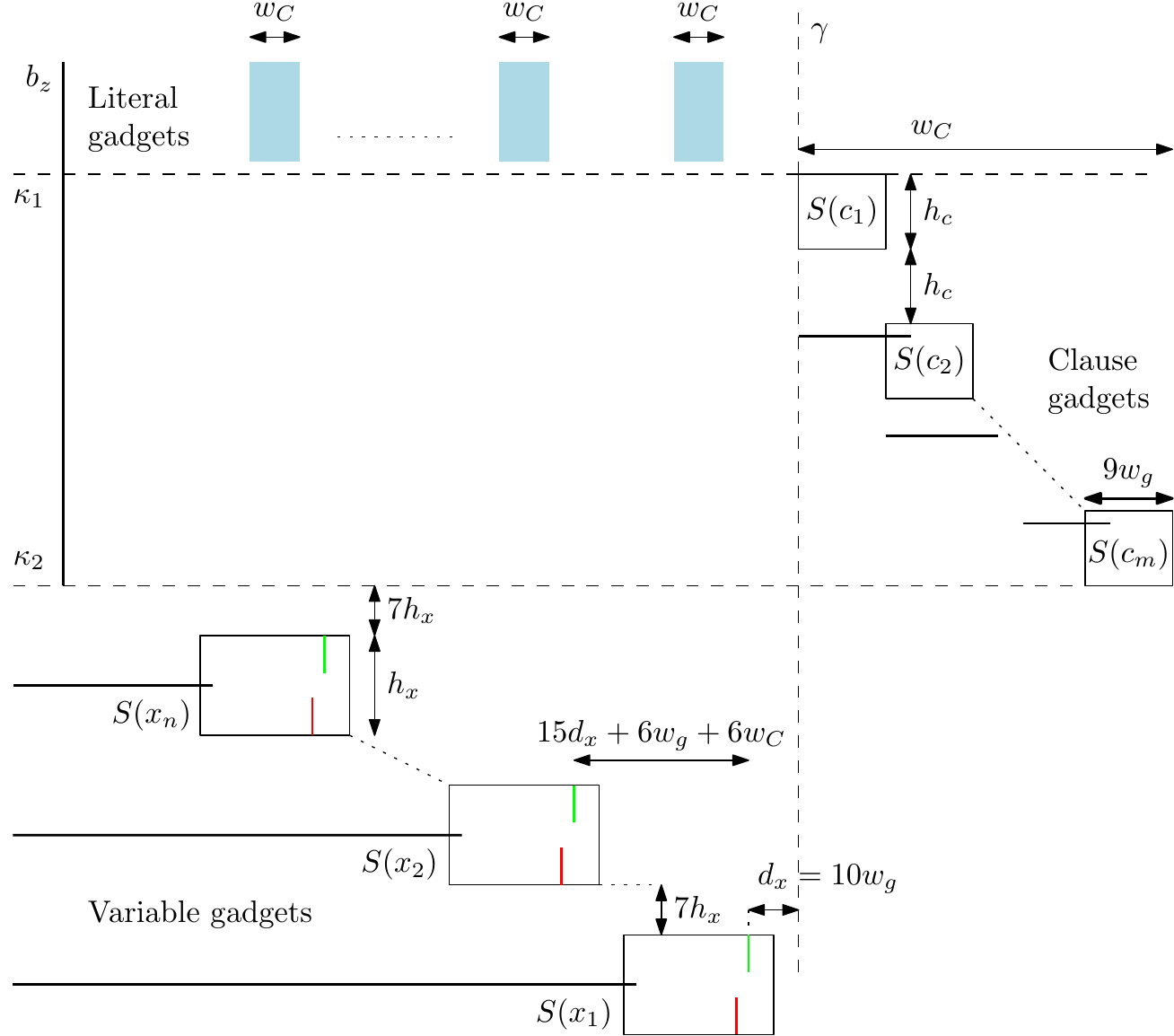}
	\caption{Overview of the reduction. Note that the lengths are not up to scale.}
	\label{fig:overview}
\end{figure*}

We first give a brief overview of the reduction. Let $F(x_1,\dots,x_n)$ be a 3SAT formula. 
As usual, our reduction will construct variable and clause gadgets using  $F(x_1,\dots,x_n)$ in a way so that any solution must first fold all variable gadgets. In each such gadget we will have the choice to fold one of two segments as the first one in the sequence, which will in turn encode the truth value of a variable. The folds themselves will shift the position of some literal segments. These segments will enable the folding of a clause gadget if and only if the truth assignment of the associated variable satisfies the clause. 
Only after all variables are folded the clause gadgets can be folded. 
This will be possible in the required number of steps if and only if for each clause gadget at least one variable was folded in a way that the corresponding literal segment was placed in the enabling part of the clause gadget.

\subsection{Global layout}
We define a vertical line $\gamma$ and two horizontal lines $\kappa_1$ and $\kappa_2$. 
We call the region between $\kappa_1$ and $\kappa_2$ and to the right of $\gamma$ the \emph{clause region}. 
The region above $\kappa_1$ and to the left of $\gamma$ is the \emph{literal region} 
and 
the region below~$\kappa_2$ and to the left of $\gamma$ is the \emph{variable region}.
As the names indicate, we will place the clause, literal, and variable gadgets in the corresponding regions. 
Figure~\ref{fig:overview} presents an illustration, 
the details will become clear in the following sections. 
Furthermore, we define $w_C$ as the sum of the width of all clause gadgets and $m_C = w_C/2$.

\subsection{Variable gadget}
For each variable $ x $ the corresponding variable gadget consists of thirteen segments. 
We call these segments \emph{true segment} $ t $, 
\emph{false segment} $ f $, 
\emph{true helper} $ t_h $, 
\emph{false helper} $ f_h $, 
\emph{true blocker one} $ t_b^1 $, 
\emph{true blocker two} $ t_b^2 $,
\emph{true blocker three} $ t_b^3 $,
\emph{false blocker one} $ f_b^1 $, 
\emph{false blocker two} $f_b^2 $, 
\emph{false blocker three} $f_b^3 $, 
\emph{next blocker one} $ b_1 $,
\emph{next blocker two} $ b_2 $,
and \emph{next blocker three} $ b_3 $. 
For a variable~$x$ we denote with $S(x)$ the set of the thirteen segments corresponding to this variable. 
The precise positioning of them is shown in Figure~\ref{fig:variable_gadget}. 

Let $w_g$ be the horizontal distance between the $t_h$ and $f_h$ segment. 
For a variable $x$, the gadget $S(x)$ without the next blocker segment $b_2$ has width $w_x = 37/2w_g + w_C$ and constant height $h_x$. 
Nearly all horizontal space is between the two helper segments and the true and false segment. 
It includes the \emph{literal strip} of variable gadget $S(x_i)$, 
see the light-blue region in Figure~\ref{fig:variable_gadget}. 
This strip has width $w_C$. 

\begin{figure*}[bt]
	\centering
	\includegraphics{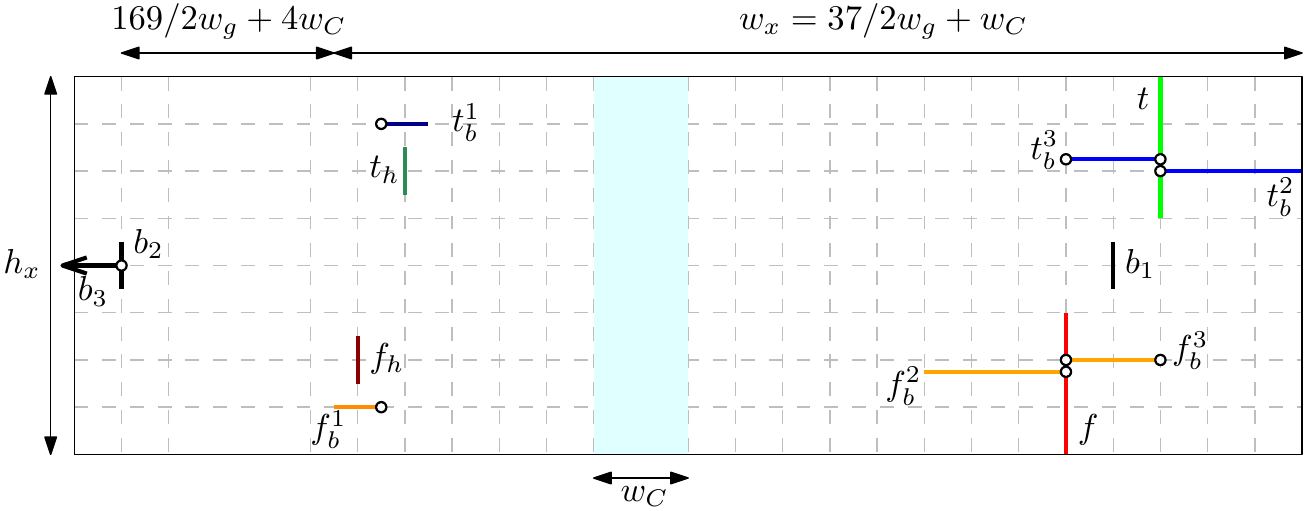}
	\caption{A variable gadget. 
		The literal strip is shown in lightblue. 
		The marked endpoints indicate that the corresponding segments are strictly separable with a vertical line. Especially, they do not intersect, but are separated by a
		small gap. Note that one cell of the lightgray grid has width and height $w_g$.   
		(The drawing is not up to scale.)}
	\label{fig:variable_gadget}
\end{figure*}

The key property of the placement of the segments is that if we start folding $t$, 
the segment $t_b^2$ is stabbed by the segment $f$, 
forcing us to fold ($t_b^1$ and) $t_h$ before $f$. 
Symmetrically, folding first along $f$ 
forces to fold ($f_b^1$ and)~$f_h$ before $t$.

It remains to explain the function of the three next blocker segments $ b_1 $,~$ b_2 $, and $ b_3 $ of a variable gadget $S(x_i)$ and the placement of $b_2$. 
These segments guarantee that no segment in the variable gadget $ S(x_{i+1}) $ can be folded before all the segments in $S(x_i)$ have been folded along. 
The segment $b_3 \in S(x_i)$ is a long horizontal segment,  
and above it lie
the segments in~$S(x_{i+1})\setminus \{b_3 \in~S(x_{i+1})\}$. 
As we will see, 
for each variable gadget $b_3$ is the last segment to be folded, 
and it prevents segments in $S(x_{i+1})$ from being folded before. 
The segment~$b_2$ is there to reset the following variable gadgets properly, i.e., it makes sure the next variable is at the same distance of $\gamma$ as the ones before. 

We now describe how the variable gadgets are placed relative to each other. 
We begin with variable $x_1$ and place it such that the true segment $t \in S(x_1)$ is $d_x = 10w_g$ units to the left of the vertical line $\gamma$ and the top-point of $t \in S(x_1)$ is $(n-1) 8h_x + 7h_x$ units below the horizontal line $\kappa_2$. 
The next variable gadgets $S(x_2),\dots , S(x_n)$ are placed to the left and above $S(x_1)$ always leaving $15d_x + 6w_g + 6w_C$ units of horizontal space between $t$ in $S(x_{i-1})$ and the $t$ segment of $S(x_i)$,  
as well as $2h_x$ units vertical space between the upper end-point of $t \in S(x_{i-1})$ and the lower end-point of $f \in S(x_i)$. 
(Hence, for each variable $x_i$ the true segment of its gadget $S(x_i)$ is placed
$10w_g + (i-1)(15d_x + 6w_g + 6w_C)$ units left of $\gamma$.)
We place $b_2 \in S(x_i)$ at $10d_x + 5w_C$ horizontally to the left of $t \in S(x_{i})$. 
Finally,~$b_3$ is placed directly to the left of $b_2$ and extended by $10w_C$ units to the left.

We now show that there are only four ways to fold a variable gadget.

\begin{lemma}\label{lem:variable_two_ways}
	Up to folding $t_b^3$ and $f_b^3$ there are exactly two ways in which a variable gadget can be folded in thirteen steps, namely:
	\begin{align*}
	& t \rightarrow t_b^1 \rightarrow t_h \rightarrow t_b^2 \rightarrow f \rightarrow f_b^1 \rightarrow f_h \rightarrow f_b^2 \rightarrow 
	b_1 \rightarrow b_2 \rightarrow b_3 \\ 
	\text{ or } & \\
	& f \rightarrow f_b^1 \rightarrow f_h \rightarrow f_b^2 \rightarrow t \rightarrow t_b^1 \rightarrow t_h \rightarrow t_b^3 \rightarrow 
	b_1 \rightarrow b_2 \rightarrow b_3.
	\end{align*}
	Furthermore, $t_b^3$ and $f_b^3$ can be folded only after $t_h$, respectively $f_h$, is folded and have to be folded before $b_1$ can be folded.
\end{lemma}
\begin{proof}
	Let $S = \{t,t_h,t_b^1,t_b^2, t_b^3,f,f_h,f_b^1,f_b^2, f_b^3,b_1, b_2, b_3\}$ be the segments of a variable gadget constructed as above. 
	Then the only three non stabbing segments are $b_2$, $t$ and $f$. 
	Notice that $b_2$ cannot be folded before $b_1$ and $b_3$, since $b_1$ would intersect $b_3$ after the folding. 
	By Observation~\ref{obs:stab} this means the only possible folds are along $t$ and $f$, respectively.
	
	Without loss of generality assume we first fold $t$. 
	Let $S'$ be the new set of segments. 
	Note that, since $t$ was not stabbing
	and no crossings were created,
	by Observation~\ref{obs:non_stab}, 
	we also did not create a new segment
	and no two segments are sharing a common supporting line. 
	Identify the segments in $S$ with their reflected counterparts in $S'$. 
	Then $f$ and $b_1$ are now stabbing $t_b^2$, but $t_b^1$ does not stab any segment anymore. 
	All other segments are still stabbing the same segments as before. 
	Hence the only possible fold is along the supporting line of~$t_b^1$. 
	By the same argumentation we get that the following folds must be in order $t_h$, $t_b^2$, $f$, $f_b^1$, $f_h$, $f_b^2$, $b_1$, and $b_3$. 
	The segment $b_2$ cannot be folded early as this would lead to $b_1$ forming a crossing with $b_3$,
	which is not allowed by Observation~\ref{obs:non_stab}.
	Finally, $t_b^3$ can never be folded before $t_h$ has been folded by Observation~\ref{obs:stab} and
	folding it at any point after that has all remaining segments of the variable gadget in one half-plane relative to the supporting line of $t_b^3$.
	Hence, the relative positions of the remaining segments in the variable gadget do not change.
	The same can be argued for $f_b^3$ relative to~$f_h$.
	We can argue in the same manner for the folding sequence starting with~$f$.
\end{proof}

Next we will show that the folding sequences of length $13n$ for the variable gadgets all fold the gadgets in order of the corresponding indices.

\begin{lemma}
	\label{lem:variable_order}
	Let $x_1,\dots,x_n$ be the variables of a 3SAT formula and consider the variable gadgets $S(x_1), \dots, S(x_n)$ as above, then the variable gadgets can be folded in $13n$ steps if and only if they are folded in order of their indices.
\end{lemma}

\begin{proof}
	The leftmost point of any segment in $S(x_1) \setminus \{b_2,b_3 \in S(x_1)\}$ and the rightmost point of any segment in $S(x_2)$ are more than $5w_C$ apart. 
	Consequently, any fold along a vertical segment in $S(x_1)$ leaves the segments in $S(x_2),\dots, S(x_n)$ on the same half-plane relative to the remaining segments in $S(x_1)$.
	For the folds along horizontal segments, the topmost point of any segment in $S(x_1)$ and the bottommost point of any segment in $S(x_2)$ are at least $7h_x$ units apart. 
	Since there are only seven horizontal segments, it is clear that no fold along a horizontal segment in $S(x_1)$ can lead to a segment in $S(x_1)$ stabbing a segment in $S(x_2) \cup \dots \cup S(x_n)$. 
	
	From Lemma~\ref{lem:variable_two_ways} it follows that folding the variable gadgets in the order of the corresponding indices gives a fold in $13 n$ steps, 
	as the sequence of 13 folds chosen for one gadget does not interfere with any other gadget.

	It remains to argue the reverse direction. 
	Assume we have a sequence of~$13n$ folds such that some segment in $s_{i+1} \in S(x_{i+1})$ is folded before some segment $s_i \in S(x_{i})$. 
	We immediately see that $s_{i+1}$ cannot be a vertical segment, since those can only be folded after $b_3 \in S(x_{i})$ was folded. 
	However, by Lemma~\ref{lem:variable_two_ways}, the first segment to be folded in $S(x_{i+1})$ is a vertical segment.
\end{proof}

\subsection{Literal gadgets}
\label{sec:literal}
As will be explained in Subsection~\ref{sec:clause} the clause gadget requires a
specific fold along a vertical line. 
These folds are going to be done along supporting lines of segments corresponding to literals.
Here, we only describe the literal gadgets and their initial placement. 
The functioning will become apparent in the then following sections.
We include a blocker segment that ensures 
all segments in variable gadgets must be folded before any segment in a literal gadget.

Let $c_j$ be a clause and $x_i$ a variable that occurs in $c_j$ and
denote with $z_{i,j}$ the corresponding literal.
For $z_{i,j}$ we create one vertical segment $z \in S(z_{i,j})$ and 
one horizontal segment  $b \in S(z_{i,j})$.
Intuitively, the former segment will be the one that enables the folding of a clause gadget if it is folded into the respective clause's good zone,
while the latter segment is needed to partially fix the order in which
the literal and clause gadgets can be folded.
In the following we call $z \in S(z_{i,j})$ the \emph{literal segment} of $z_{i,j}$ and 
$b \in S(z_{i,j})$ the \emph{literal blocker} of $z_{i,j}$.

We continue by placing the literal segments and blockers.
Roughly, for each literal $z_{i,j}$ corresponding to an occurrence of variable $x_i$ in clause $c_j$, 
the segments in $S(z_{i,j})$ are placed above $\kappa_1$ and 
inside the literal strip of the gadget $S(x_i)$
that corresponds to variable~$x_i$, see the corresponding light-blue strip in Figure~\ref{fig:variable_gadget}. 

In the following fix one literal $z_{i,j}$ and its literal segment $z \in S(z_{i,j})$ and literal blocker $b \in S(z_{i,j})$.
First, we describe the horizontal position of $z$.
We place $z$ with an offset to the true segment $t$ of~$S(x_i)$ of $(16 + \frac{1}{4})w_g + (j-1)w_c$ if $x_i$ appears positive in $c_j$ and at $(20 + \frac{3}{4})w_g + (j-1)w_c$ if $x_i$ appears negated in $c_j$. 
(This means that if $z$ is positive it is placed
\begin{align*}
    &10w_g + (i-1)(15d_x + 6w_g + 6w_C) + (16+1/4)w_g + (j-1)w_C\\
    &= (i-1)15d_x + (6i + 20 +1/4)w_g + (j+5)w_C
\end{align*}
units to the left of $\gamma$ and, similarly, if it is negative it is placed 
$(i-1)15d_x + (6i + 24 +\frac{3}{4})w_g + (j+5)w_C$
units to the left of $\gamma$.)
Finally, to avoid multiple literal segments being placed at the same $x$-coordinate, 
we give $z$ a negative horizontal offset $\delta = \frac{1}{100}w_g$ if $z_{i,j}$ if it is the literal with smallest index $i$ in the clause, 
and a positive horizontal offset of $\delta$ if the literal has the largest index $i$ in the clause.

Second, we need to position $z$ vertically and compute its length.
The bottommost point of $z$ is placed at $w_g$ units above $\kappa_1$.
The length of $z$ is given as $2j \cdot w_g$ if $z_{i,j}$ is a positive literal and
$2j \cdot w_g + w_g$ if $z_{i,j}$ is a negative literal.
It remains to describe the positioning of the literal blocker $b \in S(z_{i,j})$.
We place $b$ at $\frac{1}{2}w_g$ units above $z$ and 
its midpoint at the $x$-coordinate coinciding with the horizontal position of $z$.
Again, to avoid colinearities, we give $b \in S(z_{i,j})$'s verical position an offset $\delta$ 
depending on if $i$ is the lowest, middle, or highest index in clause $c_j$.
Finally, we set $b$ to be $\frac{1}{100}w_g$ units long.

\begin{lemma}
	\label{lem:internal_literal_order}
	Let $F(x_1,\dots,x_n)$ be a 3SAT formula and consider the variable gadgets $S(x_1), \dots, S(x_n)$  and literal gadgets $S(z_{i,j})$ as above, 
	then for $j\in \{1,\ldots,m\}$ the segments in $S(z_{i,j})$ with $i \in \{1,\ldots,n\}$ 
	are folded after all segments in $\cup_{j' < j} S(z_{i,j'})$.
	Moreover, for each $j \in \{1,\ldots,m\}$ literal gadgets $S(z_{i,j})$ corresponding to
	negative literals in $c_j$ are folded before the ones corresponding to positive literals.
\end{lemma}
\begin{proof}
	A literal segment $z \in S(z_{i,j})$ can never be folded before $b \in S(z_{i,j})$
	by Observation~\ref{obs:stab}.
	Moreover, $b$ stabs every literal segment in $\cup_{j' < j} S(z_{i,j'})$ 
	by construction.
	Hence, again by Observation~\ref{obs:stab}, $b$ cannot be folded before 
	all literal segments in $\cup_{j' < j} S(z_{i,j'})$ and their respective literal blockers are folded.
	Finally, if $z_{i,j}$ is a positive literal in clause $c_j$ its literal blocker $b$
	stabs any literal segment $z' \in S(z_{i',j})$ where $x_{i'}$ occurs negated in $c_j$.
	Consequently, by Observation~\ref{obs:stab}, the segments in $S(z_{i,j})$ are folded
	after the literal segment in $S(z_{i',j})$.
\end{proof}

To block the segments in literal gadgets to be folded before we have folded all segments in variable gadgets
we introduce one vertical segment $b_z$.
Horizontally we place it $11d_x + 5w_C$ units of horizontal distance to the left of $t \in S(x_n)$.
(Hence, $b_z$ is placed $10w_g + (n-1)(15d_x + 6w_g + 6w_C) + 11d_x + 5w_C$ 
units to the left of $\gamma$.)
Vertically we put its bottommost point on $\kappa_2$ and
extend it vertically by $(2m-1)h_c \cdot (3m + 1) \cdot w_g$ units where $h_c = \frac{59}{2}w_g$ is the height of a clause gadget
which we introduce in the following Subsection~\ref{sec:clause}.

The following observation is immediate after realizing that $b_z$ stabs the blocker $b_3$ of variable $x_n$.

\begin{obs}
	\label{obs:literal_order}
	Let $F(x_1,\dots,x_n)$ be a 3SAT formula and consider the variable gadgets $S(x_1), \dots, S(x_n)$  and literal gadgets $S(z_{i,j})$ as above, then for all literals $z_{i,j}$ the segments in $S(z_{i,j})$ are folded after all segments in $S(x_1),\ldots,S(x_n)$.
\end{obs}

\subsection{Clause Gadget}
\label{sec:clause}
For each clause $c_j$ we introduce a clause gadget
$S(c_j)$ consisting of four \emph{clause segments} $c^1,c^2,c^3$, and $c^4$ and 
one \emph{clause blocker} $b$.
Each clause gadget is confined to a rectangular area with a width $w_c$ of $9w_g$ and a height $h_c$ of $\frac{59}{2}w_g$.
In the following, we describe the positioning of all the segments in a clause gadget relative to its top left corner,
compare also \figurename~\ref{fig:clausepos}(a) and \figurename~\ref{fig:clauseneg}(a).
Segment $c^1$ has its left endpoint $\frac{9}{2}w_g$ units from the left and $4w_g$ units from the top and
its right endpoint at $6w_g$ units from the left and $\frac{5}{2}w_g$ units from the top.
The second segment $c^2$ is positioned with both its endpoints at the same $x$-coordinate of $8w_g + \frac{1}{8}w_g$ and
its upper endpoint is at zero, while its lower endpoint is also at $\frac{5}{2}w_g$ units from the top.
Segment $c^3$ has a length of $\frac{3}{2}w_g + \frac{1}{10}w_g$ to account for the offsets in the literal gadgets
an is positioned $5w_g$ units from the top with its left endpoint at $7w_g$. 
Finally, segment $c^4$ has its upper end point $\frac{11}{2}w_g$ units from the top and its lower endpoint
at $\frac{25}{2}w_g$ units from the top.

It remains to place the blocker $b \in S(c_j)$.
This segment has its right end point at $\frac{1}{4}w_g$ units from the top and $3w_g$ units from the left,
relative to clause gadget $S(c_j)$.
It then extends $12w_g$ units to the left, i.e. $9w_g$ units of $b$ lie outside of the clause gadget $S(c_j)$.

\begin{figure*}
	\centering
	\includegraphics[scale=.9,page=2]{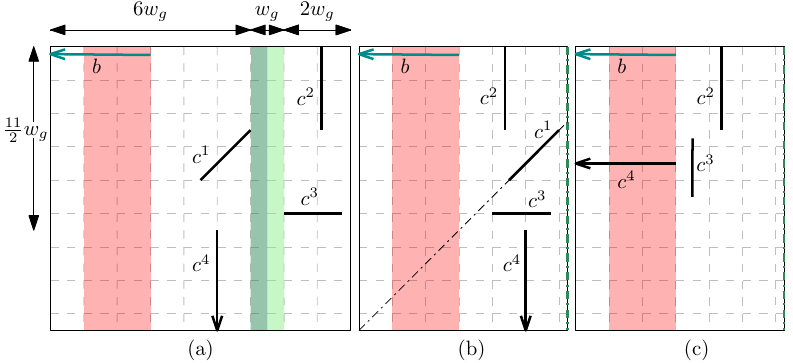}
	\caption{Clause gadget folded along a literal gadget corresponding to a positive literal.}
	\label{fig:clausepos}
\end{figure*}
\begin{figure*}
	\centering
	\includegraphics[scale=.9,page=1]{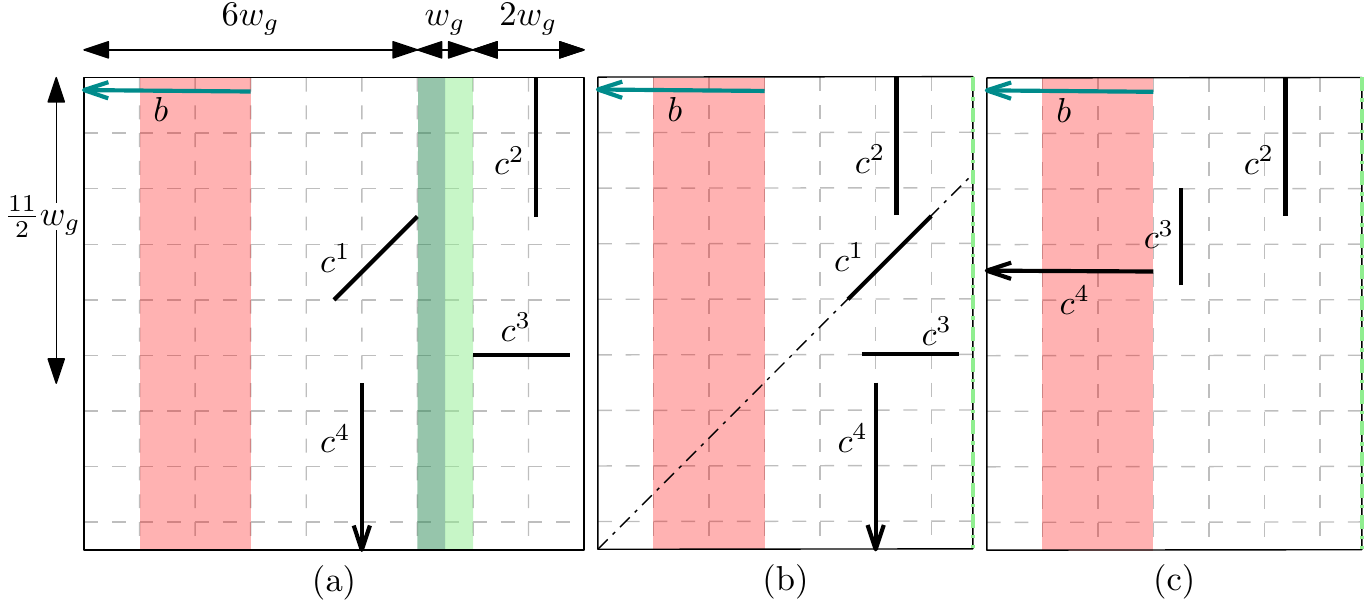}
	\caption{Clause gadget folded along a literal gadget corresponding to a negative literal.}
	\label{fig:clauseneg}
\end{figure*}

Next, we consider the possible folding sequences of a clause gadget
$S(c_j)$ for a clause $c_j$.
First, we observe that of the five segments none
can be folded without some additional sixth segment.
For $c^1$, $c^2$, $c^4$, and $b$ in $S(c_j)$ 
this can be seen immediately by Observation~\ref{obs:stab}
as they each stab another segment.
For $c^3$, observe that folding along it would produce an intersection between 
segment $c^1$ and $c^4$ which is forbidden by Observation~\ref{obs:non_stab}.

In the following we call the vertical strip in $S(c_j)$ in between $6w_g$ to $7w_g$,
but excluding the vertical lines at $6w_g$ and $7w_g$
its good zone (see the green strips in \figurename~\ref{fig:clausepos}(a) or~\ref{fig:clauseneg}(a)).
Moreover, we call the region from $w_g$ to $3w_g$ the \emph{bad zone} of the gadget.
This region is shaded in red in \figurename~\ref{fig:clausepos} and~\ref{fig:clauseneg}.
We now show under which condition it is actually possible to fold the gadget.

\begin{lemma}
	\label{lem:clause_fold}
	A clause gadget $S(c_j)$ with segments $c^1$,$c^2$,$c^3$,$c^4$, and $b$ as above,
	in a setting where no horizontal segment can be folded,
	can be folded in five steps if and only if there is a segment $s$ whose supporting line goes along its good zone.
	Moreover, in steps one and two the folds are 
	$c^1$ followed by $c^2$ and 
	the last fold is along $c^4$ or $b$.
\end{lemma}
\begin{proof}
	We first show how $S(c_j)$ can be folded if the supporting line of some segment
	lies in its good zone at either $6.25w_g$ or $6.75w_g$ (this will later be the positions
	of the literal segments introduced in Subsection~\ref{sec:literal}).
	\figurename~\ref{fig:clausepos} details the first two folds if such a vertical segment 
	is in the middle of the darkgreen strip ($6.25w_g$) and
	\figurename~\ref{fig:clauseneg} shows the same for a segment in the middle of the lightgreen strip ($6.75w_g$).
	After folding along any of these strips we next have to fold along segment $c^1 \in S(c_j)$
	by the same arguments as before.
	This leads to the configurations shown in \figurename~\ref{fig:clausepos}(c) and~\ref{fig:clauseneg}(c).
	Here, we are forced to first fold $c^2 \in S(c_j)$.
	The only other non-stabbing segment would be to fold along $c^3 \in S(c_j)$, but
	this would lead to an intersection between $c^2$ and $b \in S(c_j)$ which 
	is not allowed by Observation~\ref{obs:non_stab}.
	Finally, we fold the clause segments $c^3$ and $c^4$ of $S(c_j)$ as well as $b \in S(c_j)$
	in any order such that $c^3$ is folded before $c^4$.
	Hence, we were able to fold $S(c_j)$ in five steps as desired.
	With the same argumentation we can fold $S(c_j)$ for any other supporting line going
	through the good zone of the gadget.
	
	Now assume we could fold $S(c_j)$ without ever folding along a line going 
	through its good zone.
	Clearly if no supporting line lies in the strip between the pairs of 
	clause segments $c^1$, $c^4$ and $c^2$, $c^3$ we can never
	fold any of the clause segments.
	Moreover, if a supporting line would coincide with the vertical line at $6w_g$ or $7w_g$
	it would intersect $c^1$ or $c^3$ and hence, by Observation~\ref{obs:stab} 
	we could not fold this line.
\end{proof}

It remains to position the clause gadgets relative to each other.
Recall, that a clause has width $w_c = 9w_g$ and height $h_c = \frac{25}{2}w_g$.
Also recall, that the leftmost point of a clause gadget is $w_g$ units left of its bad zone and
the rightmost point $w_c$ units from its leftmost point.
Its topmost point is at the same point as the topmost point of its clause segment $c^2$ and
its bottommost point at distance $h_c$ from its topmost point.
Moreover, the first clause gadget $S(c_1)$ has no clause blocker.
We position $S(c_1)$ such that its leftmost point lies on $\gamma$ and 
its bottommost point is $(m-1)2h_c$ units above $\kappa_2$ and $h_c$ units below $\kappa_1$. 
The next clause $S(c_2)$ is placed $h_c$ units below the bottommost point of $S(c_1)$ and such that its leftmost point aligns with the rightmost point of $S(c_1)$.
Consequently, the clause blocker $b \in S(c_2)$ is stabbed by
the clause segments $c^2,c^4\in S(c_1)$.
Continuing in this fashion we place all $m$ clause gadgets.
Note, that the bottommost point of $S(c_m)$ precisely aligns with $\kappa_2$ and
that there is $h_c$ units of vertical space between the topmost point of $S(c_j)$ and
$S(c_{j-1})$ for any $1 < j \leq m$ and that
each clause blocker $b\in S(c_{j})$ for $1 < j \leq m$ is stabbed by 
the clause segments $c^2,c^4 \in S(c_{j-1})$.
We now state the same ordering lemma for the clauses as for the variables,
assuming that there is one vertical segment whose supporting line lies 
inside the good zone for each clause gadget.

\begin{lemma}
	\label{lem:clause_order}
	Let $F(x_1,\dots,x_n)$ be a 3SAT formula and consider the set of segments~$S$ 
	for clauses $c_j \in F(x_1,\ldots,x_n)$ as defined above and 
	assume that for each clause gadget $S(c_j)$ corresponding to a clause $c_j$
	there is exactly one vertical segment whose
	supporting line lies inside the good zone of $S(c_j)$,
	then the clause gadgets can only be folded in order $S(c_m),\ldots,S(c_1)$.
\end{lemma}
\begin{proof}
	As just observed the clause blocker $b\in S(c_j)$ is stabbed by the clause segments
	$c^2,c^4\in S(c_{j-1})$ for every $1 < j \leq m$.
	Moreover, this means that the vertical segment in the good zone of each $S(c_{j-1})$
	stabs $b \in S(c_j)$.
	Hence, with Lemma~\ref{lem:clause_fold}, 
	the only segment we can fold in step one is the vertical segment
	whose supporting line lies in the good zone of $S(c_m)$.
	Crucially, also by Lemma~\ref{lem:clause_fold},
	the last segment we fold is either $b$ or $c^4$.
	Moreover, after the fold along $c^1$ segment $c^4$ is
	at the same $x$-coordinates as $b$ and hence any vertical segment that
	was stabbing $b$ is now also stabbing $c^4$.
	Consequently, we can only fold $S(c_{m-1})$ once all segments in $S(c_m)$ are folded.
	
	To complete the proof observe that all segments in any $S(c_j)$ for $1 \leq j < m$ are
	completely in one half-plane relative to the supporting lines of segments in $S(c_m)$
	when we fold along them.
	Hence, their relative positions did not change.
\end{proof}

\subsection{Correctness} It remains to argue the correctness of our reduction.
Theorem~\ref{thm:restricted} then follows directly from Lemma~\ref{lem:correctness}.
\begin{lemma}
	\label{lem:variable_undo}
	Given a 3SAT formula $F(x_1,\ldots ,x_n)$, 
	let $S(x_1), \dots, S(x_n)$ be the variable gadgets constructed as above. 
	Then, after folding $S(x_1),\dots,S(x_i)$, the horizontal distance between $t \in S(x_{i+1})$ and $\gamma$ is $d_x$.
\end{lemma}

\begin{proof}
	First consider the fold along the segments in $S(x_1)$.
	By Lemma~\ref{lem:variable_two_ways} we know there are four sequences of folds we have to consider.
	We first show that the statement initially holds for folding all segments in $S(x_1)$. 
	
	Let $h(s,t)$ be the horizontal distance between the two vertical segments or lines $s$ and $t$, and
	let $h'(s,t)$ be the horizontal distance between segments or lines $s$ and $t$ after $S(x_1)\setminus \{b_2,b_3\}$ was folded.
	
	We start by showing that, after folding all segments in $S(x_1)\setminus \{b_2,b_3\}$, $h'(b_2,\gamma) = 10w_g + w_C$.
	Note that for the horizontal distance the folds along horizontal segments do not  matter.
	Hence, it suffices to consider a sequence of folds starting with $t \in S(x_1)$ and one starting with $f \in S(x_1)$.
	
	We start with $t \in S(x_1)$, and we use that $h(t,b_1) = 1w_g$: 
	\begin{align*}
	h'(b_2,\gamma)= h(b_2,\gamma) &- 2h(t,\gamma) \\ &- 2(h(t_h,\gamma) - 2h(t,\gamma)) \\ &- 2(h(f,\gamma) - 4w_g)\\ &- 2(h(f_h,\gamma) - 2h(t,\gamma) - 2w_g)\\ &- 2(h(b_1,\gamma) - 4w_g).
	\end{align*}
	
	Since $h(t,\gamma) = 10w_g$ and $h(f,\gamma) = 12w_g$, we get
	\begin{align*}
	h'(b_2,\gamma)= h(b_2,\gamma) &- 20w_g \\ &- 2h(t_h,\gamma) + 40w_g \\ &- 16w_g\\ &- 2h(f_h,\gamma) +44w_g\\ &- 2h(b_1,\gamma) +8w_g\\
	= h(b_2,\gamma) &- 2h(t_h,\gamma) - 2h(f_h,\gamma) \\  &- 2h(b_1,\gamma) + 56w_g.
	\end{align*}
	
	For the sequence starting with $f \in S(x_1)$
	\begin{align*}
	h'(b_2,\gamma) = h(b_2,\gamma) &- 2h(f,\gamma) \\ &- 2(h(f_h,\gamma) - 2h(f,\gamma)) \\ &- 2(h(t,\gamma))\\ &- 2(h(t_h,\gamma) - 2h(f,\gamma))\\ &- 2(h(b_1,\gamma) - 2w_g).
	\end{align*}
	
	Since $h(t,\gamma) = 10w_g$ and $h(f,\gamma) = 12w_g$, we get
	\begin{align*}
	h'(b_2,\gamma) = h(b_2,\gamma) &- 24w_g \\ &- 2h(f_h,\gamma) + 48w_g \\ &- 20w_g\\ &- 2h(t_h,\gamma) + 48w_g\\ &- 2(h(b_1,\gamma) +4w_g\\
	= h(b_2,\gamma) &- 2h(t_h,\gamma) - 2h(f_h,\gamma) - 2h(b_1,\gamma) + 56w_g.
	\end{align*}
	
	It follows that for both sequences the distance between $b_2$ and $\gamma$ will be equal. 
	Furthermore, we can compute this distance. 
	Since $h(b_1,\gamma) = 11w_g$,  $h(f_h,\gamma) = h(t_h,\gamma) + w_g$, 
	and $h(b_2,\gamma) = 11d_x + 5w_c$ arrive to
	\begin{align*}
	h'(b_2,\gamma) &= h(b_2,\gamma) - 4h(t_h,\gamma) + 32w_g\\
	&= 11d_x + 5w_C - 4(d_x + h(t,t_h)) + 32w_g \\
	&= 10d_x + 2w_g + 5w_C - 4(14 + w_C) \\
	&= 4d_x + 6w_g + w_C.
	\end{align*}
	
	To get the lemma, note that since the true segment $t \in S(x_2)$ and $b_2$ were always in the same half-plane in the
	folding sequences above, the distance between $t \in S(x_2)$ and $b_2$ did not change. 
	Hence, the distance between $t \in S(x_2)$ and $\gamma$ after folding along all segments in $S(x_1)$ is \begin{align*} 
	h(b_2,t) - h'(b_2,\gamma) &= 5d_x + 6w_g + w_C - (4d_x + 6w_g + w_C) = d_x. 
	\end{align*}
	Since by Lemma~\ref{lem:variable_order} the variable gadgets $S(x_1),\dots,S(x_n)$ can only be folded in order of
	their indices the above calculation holds for any $S(x_i)$ and $S(x_{i+1})$ in such a sequence.
\end{proof}

\begin{lemma}
	\label{lem:positive_literal}
	
	Given a 3SAT formula $F(x_1,\ldots ,x_n)$, 
	let $S$ be the set of segments constructed by the reduction. 
	Let $\mathcal S$ be a folding sequence of the variable gadgets up to $S(x_{i-1})$.
	Let $S'$ be the segments after a fold along $s \in \{t,f\} \subset S(x_i)$.
	If $s = t$, then for every literal gadget $S(z_{i,j})$ with clause $c_j \in F(x_1,\dots,x_n)$ we have that $S'(z_{i,j})$ is in the good zone of $S'(c_j)$ 
	if $z_{i,j}$ is a positive literal, and in the bad zone otherwise. 
	If instead $s = f$, 
	then for every literal gadget $S(z_{i,j})$ with clause $c_j \in F(x_1,\dots,x_n)$ we have that $S'(z_{i,j})$ is in the good zone of $S'(c_j)$ 
	if $z_{i,j}$ is a negative literal, and in the bad zone of $S'(c_{j+1})$ otherwise. 
	Moreover, it is the only literal segment reflected to this coordinate.
\end{lemma}

\begin{proof}
	For the proof we only need to consider horizontal distances, 
	and thus, folds along vertical segments.
	Let us first not consider the $\delta$ offsets. 
	By Lemma~\ref{lem:variable_undo} we know that the distance between the clauses and $t \in S(x_i)$ remains $d_x = 10w_g$. 
	
	We have to consider four cases. First let $S'$ be the set of segments after a fold along $t \in S(x_i)$. 
	Then, any positive literal segment $S(z_{i,j})$ is reflected to a distance $(16 + \frac{1}{4})w_g + (j-1)w_c$ from $t$. 
	At the same time the good zone of $c_j$ starts at distance $16w_g + (j-1)w_c$ from $t$ and spans $w_g$ units. 
	Hence the segment is reflected into the good zone as desired.
	For a negative literal segment $S(z_{i,j})$ is reflected to a distance 
	\begin{align*}
	\left(20 + \frac{3}{4}\right)w_g + (j-1)w_c  &= \left(20 + \frac{3}{4}\right)w_g + j \cdot w_c - 9w_g \\ &= \left(11 + \frac{3}{4}\right)w_g + j\cdot w_c
	\end{align*} from $t$. 
	At the same time the bad zone of $c_{j+1}$ starts at distance $10w_g + j\cdot w_c $ from $t$ and spans $4w_g$ units.
	Hence the segment is reflected into the bad zone of the next clause gadget.
	
	Now consider a fold along $f \in S(x_i)$ segment. 
	Then for a positive literal $z_{i,j}$ the literal segment of $S(z_{i,j})$ is reflected to $(16 + \frac{1}{4})w_g + (j-1)w_c - 4w_g = (12 + \frac{1}{4})w_g + (j-1)w_c$ to the right of $t \in S(x_i)$ and the bad zone of $S(c_j)$ starts at distance $10w_g + (j-1)w_c$ from $t$ and spans $4w_g$ units. 
	As desired the segment is reflected to the bad zone.
	Finally consider a negative literal $z_{i,j}$, then the literal segment of $S(z_{i,j})$ is reflected to $(16 + \frac{3}{4})w_g + (j-1)w_c $ to the right of $t \in S(x_i)$ and the good zone of $S(c_{j})$ starts at distance $16w_g + (j-1)w_c$ from $t$ and spans $w_g$ units. 
	As desired, the segment is reflected to the good zone of $S(c_j)$.
	
	So far all literal segments belonging to the same clause $c_j$ are reflected to the same coordinates in the good and bad zones. 
	Now consider the $\delta$ offset we added depending on the index of the variable. 
	Since we can assume that no variable appears twice in a clause the indices $i$ of all three literals are different. 
	Hence, to exactly one we add $\delta$, to one we do not add or subtract it and from one we subtract $\delta$. 
	Since $\delta$ is greater zero, but smaller than $\frac{1}{4}w_g$, 
	we get that the three literal segments for a clause $c_j$ are still inside the good and bad zone, but never at the same coordinate.
\end{proof}

The reduction can be computed in polynomial time in $n$: 
for a given 3SAT formula $F(x_1,\ldots,x_n)$, we can construct the set $S$ of 
segments in polynomial time in $n$
using the specified coordinates (if we want integer coordinates we can set $w_g=400$).
Theorem~\ref{thm:restricted} follows from the next lemma.

\begin{lemma}\label{lem:correctness}
	Let $F(x_1,\dots,x_n)$ be a formula of a 3SAT instance and
	$S$ the set of segments constructed from $F$ as above, then $F$ is satisfiable if and only if $S$ can be folded in $\msize{S}$ steps. 
\end{lemma}
\begin{proof}
	Let $S$ be the set of segments constructed as above from a 3SAT formula $F(x_1,\dots,x_n)$, and let $\mathcal S$ be a folding sequence of length $\msize{S}$, folding $S$. 
	Combining Lemmas~\ref{lem:variable_order}, \ref{lem:clause_fold} and \ref{lem:clause_order}, as well as Observation~\ref{obs:literal_order}, we get that in $S$ the variable gadgets must all be folded first, then the literal and clause gadgets. 
	Furthermore, by Lemma~\ref{lem:variable_two_ways}, each fold of a variable gadget starts with the $t$ or~$f$ segment for each variable. 
	By Lemma~\ref{lem:positive_literal} we know that a literal segment gets reflected into the good zone of a clause if and only if it appears positive and the $t$ segment is folded first or it appears negated and the $f$ segment of the corresponding variable is folded first. 
	Additionally, by Lemma~\ref{lem:clause_fold}, no clause gadget can be folded in the required number of steps if there is no segment inside its good zone. 
	Moreover, by Lemma~\ref{lem:internal_literal_order} the literal segments corresponding to literals of one clause $c_j$
	have to be folded before any literal segment corresponding to a literal in clause $c_{j-1}$.
	Since by Lemma~\ref{lem:clause_fold} the last segment when folding a clause is $b$ or the, by then, parallel segment $c^4$
	we find that the clause and literal gadgets are folded ordered by the index of their corresponding clause from $m$ to $1$.
	Finally observe that in such a folding sequence no two segments ever cross.

	Two technicalities remain, for which the two transformations done to the 3SAT formula are crucial.
	First, as the last clause $c_m$ is a positive one, it cannot happen that a literal segment is ever reflected to a position that
	would be the bad zone of clause $c_{m+1}$.
	Second, we can only partially fix the order in which the literal segments of one clause are folded.
	Observe, that if $|S|$ folds are sufficient there exists a sequence of folds along literal segments such that either
	the segment we fold is the one that breaks the cyclic dependencies in the clause gadget or
	the segment is on the convex hull.
	Hence, for the latter discussed reverse direction in which we are given a satisfying assignment,
	a folding sequence in $|S|$ steps is easy to construct.
	It remains to proof that we can never fold along a segment such that two literal segments are reflected ontop of each other.
	Here, the transformation by Lemma~\ref{lem:transform} makes the argument much simpler,
	as we only have to consider clauses containing two negative or two positive literals.
	
	We begin by assuming that the next clause gadget that we can fold is $S(c_j)$ and
	consider the two cases if the first literal segment folded that lies in the good zone of $S(c_j)$ corresponds to a positive or a negative literal.
	In the following we measure distances from the rightmost point of $S(c_{j-1})$,
	hence the bad zone of $S(c_j)$ always stars at at $w_g$ units.
	First, assume this first literal segment $z$ corresponds to a positive literal.
	Then, there is at most one other literal segment remaining as
	literal segments corresponding to negative literals are folded before 
	the ones corresponding to positive literals by Lemma~\ref{lem:internal_literal_order} and 
	no clause has three positive literals by Lemma~\ref{lem:transform}.
	Let $z'$ be this remaining literal segment corresponding to a positive literal.
	If $z'$ was reflected to the bad zone of $S(c_j)$ no alignment can happen as $z'$ is only folded after all segments in $S(c_j)$ are folded.
	In the case that $z'$ was initially reflected to the good zone of $S(c_j)$ 
	it is in the following reflected at most along the segments $c^2 \in S(c_j)$.
	After the fold along $c^2$ the segment $z'$ is at horizontal position $(2 + \frac{1}{2})w_g$.
	By Lemma~\ref{lem:positive_literal} the other literal segments in the bad zone of $S(c_j)$ are at $(1 + \frac{3}{4})w_g$.
	Hence, $z'$ is on the convex hull and can be folded, by Observation~\ref{obs:ch}, before any following literal segment as desired.
	
	Second, assume the first literal segment $z$ corresponds to a negative literal.
	Here we have to consider at most two more existing literal segments.
	Assume that only one more literal segment $z'$ exists.
	In case $z'$ corresponds to a positive literal we are done as above.
	If $z'$ corresponds to a negative literal  we can nearly argue as before, with the difference that $z_n$ is reflected along $c^2$ and $c^3$ of $S(c_j)$
	to be reflected to position $(2+\frac{1}{2})w_g$.
	Note, that it could also be folded before $c^3$ in which case $z_n$ is just on the convex hull of the remaining segments.
	In the following assume $c_j$ has three literals.
	Then, two cases are to be considered,
	either the two remaining literal segments both correspond to positive literals or
	to a positive and a negative literal.
	We begin with the latter and let $z_p$ and $z_n$ be the literal segments
	corresponding to a positive and a negative literal of $c_j$ respectively.
	If $z_n$ was initially reflected to the bad zone of $S(c_{j+1})$ we get that
	after folding at $z$ the segment $z_n$ is at $(2+\frac{3}{4})w_g$.
	Since $z_n$ has to be folded before $z_p$ this is not a problem if $z_p$ lies in the bad zone of $S(c_j)$.
	In case $z_p$ lies in the good zone though it will be reflected to $-\frac{3}{4}w_g$ after folding along $z_n$.
	Consequently, $z_p$ will then stab the segment $c^3 \in S(c_{j-1})$ and we cannot proceed as 
	no literal segment in the good zone of $S(c_{j-1})$ can be folded before $z_p$ is folded.
	Hence, $z_n$ had to be folded before $z$.
	Now, consider that $z_n$ was initially also in the good zone of $S(c_j)$, 
	but then after folding along $z$ we find that $z_n$ is on the convex hull of the remaining segments and 
	furthermore it has to be folded before $z_p$ by Lemma~\ref{lem:internal_literal_order}.
	The remaining segment $z_p$ can then be treated as above in the case of just one further literal segment.
	Lastly, we consider the case that both remaining literal segments correspond to positive literals in $c_j$ and
	let $z_1$ and $z_2$ be the two literal segments.
	If both are initially in the bad zone of $S(c_j)$ no alignments are possible as they are both only folded after all segments in $S(c_j)$.
	If both are initially in the good zone of $S(c_j)$ we can treat them as if only one literal segment corresponding to a positive literal was remaining.
	Finally, assume that, w.l.o.g., $z_1$ was initially reflected to the good zone of $S(c_j)$ and $z_2$ to the bad zone.
	Now, $z_1$ is either folded before $z_2$ or is as above reflected along $c^2 \in S(c_j)$ and $c^3 \in S(c_j)$
	to $(2 + \frac{1}{2})w_g$.
	The only remaining case is that $z_2$ is folded before $z_1$, but then 
	similarly to above $z_1$ is reflected to a position of $-(1+\frac{3}{4})w_g$ and 
	intersects the segment $c^3\in S(c_j)$.	
	As a result, we find a satisfying assignment of the variables of $F$ by simply setting $x_i$ to true if in $\mathcal S$ the true segment~$t \in S(x_i)$ was folded before $f\in S(x_i)$ and to false if the converse holds.

	For the reverse direction, let $F(x_1,\dots, x_n)$ be again a 3SAT formula and~$\mathcal F$ a fulfilling assignment. 
	Now let $S$ be the set of segments constructed as above. 
	We construct a folding sequence $\mathcal S$ as follows. 
	For every variable gadget $S(x_i)$ pick $t \in S(x_i)$ to be folded first if $x_i$ is set to true in $\mathcal F$ and pick $f \in S(x_i)$ otherwise. 
	It follows from Lemmas~\ref{lem:variable_order},~\ref{lem:internal_literal_order},~\ref{lem:clause_fold}, and~\ref{lem:clause_order}, and Observation~\ref{obs:literal_order}, that this fixes the folding sequence. 
	Assume that $\msize{\mathcal S} \neq \msize{S}$ and there is a fold in $\mathcal S$ which splits another segment, aligns them, or leads to a crossing. 
	By the way we picked $\mathcal S$ none of these three cases can occur when we fold the~$n$ variable gadgets $S(x_1),\dots, S(x_n)$.
	Hence, the fold splitting, aligning, or crossing two segments must happen when the clause gadgets are folded, but then, with Lemma~\ref{lem:clause_fold}, this can only happen in a clause $c_j$, for which none of the corresponding literal segments $S(z_{i,j})$ was in the good zone of $S(c_j)$. 
	However, since $\mathcal F$ is satisfying, this cannot be the case by Lemma~\ref{lem:positive_literal}. 
\end{proof}

\subsection{Modifying the construction}\label{sec:gp}

We showed that deciding if there is a solution to the restricted segment folding problem for an instance $S$ with at most $|S|$ folds 
is NP-hard. 
Recall that in this setting we do not allow folds along stabbing segments or folds that produce crossings. 
The reduction presented above uses  perpendicular directions 
and constructs an instance that is not in general position, so we cannot directly remove the folding restrictions. 

If we look at the coordinates of the endpoints of $|S|$ segments in $\mathbb{R}^{4|S|}$,
the points corresponding to configurations that are not in general position form a set of measure zero.
Thus, small random perturbations would \emph{almost surely} ensure 
that a modified instance is in general position (even after any $|S|$ folds). 
We next show how to deterministically perturb an instance $S$ produced by the reduction presented above so that a folding sequence of length $|S|$ solving the segment 
folding problem does never fold along stabbing segments or make folds that produce crossings.

We perturb the directions of all the segments in the construction by a very small amount that does not change the desired behaviour of the reduction but prevents that two or more segments are folded simultaneously in an optimal solution. 
We encode the segments as a tuple consisting of an endpoint, the length, and an angle. 
Setting $w_g=400$ in the above reduction produces a set of segments whose endpoints have integer coordinates. 
Thus, to encode an endpoint and the length we can use the values of the construction, 
with one exception: 
the only segments that are not aligned in the vertical or in the horizontal directions 
are the $c^1$ segments in the clause gadgets. 
To encode them we use the (scaled) integer coordinates of their rightmost endpoint 
and set the length to $2w_g=800$ (instead of $\frac{3}{2}\sqrt{2}w_g$). 
This change does not affect the behaviour of the clause gadget. 

For the angles, we include small perturbations in the following manner. 
We use $O(1)+ (|S|+1)^2$ bits (we use a binary encoding of the angles, thus showing weakly NP-hardness). 
The most significant bits encode the direction defined in the construction above, 
that is, one of the eight possible angles. 
The remaining $(|S|+1)^2$ bits encode the small perturbations. 
We divide the $(|S|+1)^2$ bits into $|S|$ groups of $|S|+1$ consecutive bits each 
and assign a group to each segment. 
For a segment, the perturbation of the angle consists of setting the rightmost bit of its group to~$1$ 
(and the remaining $(|S|+1)^2-1$ bits are set to $0$). 
When folding along a segment $s$, all reflected segments keep their perturbation bits intact, 
except that in the group corresponding to $s$ they get new bits (it gets summed the bits in that region for $s$ shifted one position to the left). 
Note that by folding $|S|$ times the leading $1$ bit in a group can move at most $|S|$ positions to the left and thus it cannot overflow into another group. 
In particular, this implies that two segments cannot be aligned unless both of them (or parts of them) have been folded before. 
Thus, if there is a folding sequence of length $|S|$ it must fold each segment exactly once. 
This completes the modification of the reduction and the proof of Theorem~\ref{thm:nphardness}. 

\section{Conclusions}

We have proved that given a set of segments finding the sequence minimizing the number of folds that solves the segment folding problem is weakly NP-hard. 
We also showed that the restricted segment folding problem, in which only folds along lines that do not intersect the interior of segments and do not produce crossings are allowed, 
is strongly NP-hard. 
Though optimizing the number of folds in these problems is hard, 
it might be possible that all instances of the segment folding problem can be solved in a polynomial number of folds in the number of segments. 
In Figure~\ref{fig:inf} we presented an instance in which an infinite sequence of folds is possible, but the shortest sequence of folds is finite. 
It is an interesting open problem to find upper bounds on the length of the shortest sequence of folds in terms of the number of segments in the instance, or to show that this number cannot be bounded.  

We conjecture that the segment folding problem is strongly NP-hard. 
To prove this using our reduction we would need a deterministic perturbation of any instance $S$ that admits a polynomial encoding in unary. 
We note that the vertical separation between gadgets can be increased without altering the correctness proof of the reduction. 
However, the horizontal distances are more delicate, since we need to carefully land in the good and bad zones of clause gadgets. 
For that we use a precise reset of each variable gadget. 
Thus, to keep the coordinates small, instead of relying on a deterministic perturbation the endpoints of the segments, we could make copies of all the segments. 
One perturbation that we can deterministically do without blowing the coordinate sizes more than a polynomial in $|S|$, is to alter the length of all the segments so that no two of them have the same length. 
Then a possible strategy would be to make copies of each segment in $S$ close to it but slightly tilted, so that folding along a segment stabbing a segment in $S$ and its copies produces too many non-perpendicular crossings.

Finally, the segment folding problem admits a natural generalization in which folds along arbitrary lines are allowed:
\begin{center}
\begin{boxedminipage}{0.98 \columnwidth}
\textsc{Generalized Segment Folding Problem} \\ [5pt]
\begin{tabular}{l p{0.80 \columnwidth}}
Input: & A set $S$ of line segments $s_1,s_2,\ldots,s_{|S|}$ in the plane and an integer $\kappa$.\\
Operation: & Fold along any line.\\
Question: & Is there a sequence of $\kappa$ operations such that all (parts of) segments $s_i$ become folding lines?
\end{tabular}
\end{boxedminipage}
\end{center}
The computational complexity of this problem as well as bounds on the number of folds needed to solve any instance are other interesting open questions. 

\bibliographystyle{plainurl}
\bibliography{references_folding}

\begin{thebibliography}{1}

\bibitem{ADK17}
Hugo~A. Akitaya, Erik~D. Demaine, and Jason~S. Ku.
\newblock Simple folding is really hard.
\newblock {\em Journal of Information Processing}, 25:580--589, 2017.

\bibitem{DBLP:conf/wg/ArroyoKPSVW20}
Alan Arroyo, Fabian Klute, Irene Parada, Raimund Seidel, Birgit Vogtenhuber,
  and Tilo Wiedera.
\newblock Inserting one edge into a simple drawing is hard.
\newblock In Isolde Adler and Haiko M{\"{u}}ller, editors, {\em Proceedings of
  the 46th International Workshop on Graph-Theoretic Concepts in Computer
  Science ({WG}'20)}, volume 12301 of {\em LNCS}, pages 325--338. Springer,
  2020.
\newblock \href {https://doi.org/10.1007/978-3-030-60440-0_26}
  {\path{doi:10.1007/978-3-030-60440-0_26}}.

\bibitem{OriRobo}
Devin Balkcom.
\newblock Origami robot folding a paper hat.
\newblock \url{https://www.youtube.com/watch?v=djPdzCj4k14}, 2004.
\newblock Accessed: 2019-06-19.

\bibitem{Chan}
Brian Chan.
\newblock The making of mens et manus (in origami), vol. 1.
\newblock
  \url{http://techtv.mit.edu/collections/chosetec/videos/361-the-making-of-mens-et-manus-in-origami-vol-1},
  2007.
\newblock Accessed: 2019-06-19.

\end{thebibliography}

\end{document}